\documentclass[11pt]{article}
\usepackage{amsmath,amsfonts,amsthm,amssymb,ifthen}

\usepackage[utf8]{inputenc} 
\usepackage[T1]{fontenc}    
\usepackage{url}            
\usepackage{booktabs}       
\usepackage{amsfonts}       
\usepackage{nicefrac}       
\usepackage{microtype}      
\usepackage{amsmath}        
\usepackage{amsthm}         
\usepackage{mathtools}      
\usepackage{float}
\usepackage{enumitem}
\usepackage{xspace}
\usepackage{fullpage}
\usepackage[dvipsnames]{xcolor}
\usepackage{hyperref}
\usepackage{pdfsync}

\usepackage{libertine}
\usepackage{libertinust1math}
\usepackage{inconsolata}
\usepackage{algorithmic,algorithm}

\usepackage{graphicx} 
\usepackage{subcaption} 

\usepackage[noend,ruled,linesnumbered,algo2e]{algorithm2e}

\usepackage{bbm}
\usepackage{wrapfig}  
\usepackage[margin=1in]{geometry}

\newcommand{\declarecolor}[2]{\definecolor{#1}{RGB}{#2}\expandafter\newcommand\csname #1\endcsname[1]{\textcolor{#1}{##1}}}
\declarecolor{White}{255, 255, 255}
\declarecolor{Black}{0, 0, 0}
\declarecolor{LightGray}{216, 216, 216}
\declarecolor{Gray}{127, 127, 127}
\declarecolor{Orange}{237, 125, 49}
\declarecolor{LightOrange}{251,229, 214}
\declarecolor{Yellow}{255, 192, 0}
\declarecolor{LightYellow}{255, 242, 200}
\declarecolor{Blue}{91, 155, 213}
\declarecolor{LightBlue}{222, 235, 247}
\declarecolor{Green}{112, 173, 71}
\declarecolor{LightGreen}{226, 240, 217}
\declarecolor{Navy}{68, 114, 196}
\declarecolor{LightNavy}{218, 227, 243}
\declarecolor{DarkBlue}{0,20,115}

\hypersetup{
	colorlinks=true,
	pdfpagemode=UseNone,
	citecolor=Navy,
	linkcolor=Navy,
	urlcolor=Navy,
}

\def\prob#1#2{\mbox{Pr}_{#1}\left[ #2 \right]}

\def\defeq{\stackrel{\mathrm{def}}{=}}

\def\norm#1{\left\| #1 \right\|}

\newtheorem{theorem}{Theorem}
\newtheorem{lemma}{Lemma}[section]
\newtheorem{corollary}{Corollary}[section]

\newtheorem{proposition}{Proposition}
\newtheorem{definition}{Definition}[section]

\newtheorem{fact}{Fact}

\newtheorem*{remark}{Remark}

\def\calX{\mathcal{X}}
\def\calY{\mathcal{Y}}

\newcommand\R{\boldsymbol{\mathbb{R}}}

\newcommand\xx{\boldsymbol{\mathit{x}}}

\newcommand{\diffp}{\varepsilon}

\newcommand{\cR}{\mathcal{R}}

\newcommand{\cX}{\mathcal{X}}

\newcommand{\lap}{\texttt{Lap}}

\newcommand{\expo}{\texttt{Expo}}

\newcommand{\len}[1]{\texttt{len}_{#1}}

\newcommand{\plm}{M_{\texttt{plm}}}
\newcommand{\qplm}[2]{q_\texttt{plm}({#1};{#2})}
\newcommand{\elly}[2]{\ell({#1};{#2})}
\newcommand{\marginal}[2]{\Delta({#1};{#2})}

\usepackage{parskip}

\usepackage{pythonhighlight}

\begin{document}

\title{Unifying Laplace Mechanism with Instance Optimality in Differential Privacy}

\author{  David Durfee\\
  Anonym by Mozilla\\
  \texttt{ddurfee@mozilla.com}}

\maketitle
\thispagestyle{empty}

\begin{abstract}

We adapt the canonical Laplace mechanism, widely used in differentially private data analysis, to achieve near instance optimality with respect to the hardness of the underlying dataset.
In particular, we construct a piecewise Laplace distribution
whereby we defy traditional assumptions and show that Laplace noise can in fact be drawn proportional to the local sensitivity when done in a piecewise manner.
While it may initially seem counterintuitive that this satisfies (pure) differential privacy and can be sampled, we provide both through a simple connection to the 
exponential mechanism and inverse sensitivity
along with the fact that the Laplace distribution is a two-sided exponential distribution.
As a result, we prove that in the continuous setting our \textit{piecewise Laplace mechanism} strictly dominates the inverse sensitivity mechanism, which was previously shown to both be nearly instance optimal and uniformly outperform the smooth sensitivity framework.
Furthermore, in the worst-case where all local sensitivities equal the global sensitivity, our method simply reduces to a Laplace mechanism.
We also complement this with an approximate local sensitivity variant to potentially ease the computational cost, which can also extend to higher dimensions.

\end{abstract}

\section{Introduction}\label{sec:intro}

At the inception of differential privacy, the fundamental Laplace mechanism was introduced to release private estimates of function outputs  \cite{dwork2006calibrating}. 
This mechanism determines the maximum amount one individual could plausibly change the function output for any possible dataset and adds Laplace noise proportional to that quantity (the \textit{global sensitivity}). 
The Laplace distribution is particularly well-suited as a noise addition because it can be considered tight with respect to the  differential privacy definition.
As such, it is still one of the most widely used mechanisms both directly for more basic data analytics and as a component for more complex data modeling.

While the Laplace distribution can be considered tight, the global sensitivity can often be significantly greater than the amount one individual could change the outcome from the underlying dataset (the \textit{local sensitivity}).
Adding Laplace noise proportional to the local sensitivity would then provide the ideal mechanism, but unfortunately it is known that this does not allow for differential privacy guarantees.

Given the widespread applicability of the Laplace mechanism, several commonly used frameworks were promptly introduced both explicitly with
\textit{smooth sensitivity} \cite{nissim2007smooth} and \textit{propose-test-release} \cite{dwork2009differential},
and implicitly with
\textit{inverse sensitivity mechanism} \cite{mcsherry2007mechanism},
to instead adapt to the hardness of the underlying data. 
These frameworks apply generally 
by utilizing the local sensitivity through more intricate techniques
to improve utility while maintaining privacy.
Their applications include computing median \cite{nissim2007smooth, smith2011privacy}, mean estimation \cite{bun2019average, hopkins2023robustness}, covariance of Gaussians \cite{hopkins2023robustness}, linear regression \cite{dwork2009differential, asi2020instance}, principal component analysis \cite{gonem2018smooth, asi2020instance}, high-dimensional regression problems \cite{thakurta2013differentially}, outlier analysis \cite{okada2015differentially}, convex optimization \cite{asi2021adapting}, and graph data \cite{kasiviswanathan2013analyzing, ullman2019efficiently}.

\subsection{Our techniques}\label{subsec:our_techniques}

In this work we improve upon the previous frameworks by defying the assumption that Laplace noise cannot be added proportional to the local sensitivity.
Specifically, we show that this is possible when done in a piecewise manner.
Our goal is to construct a piecewise Laplace distribution such that the scale parameter in the PDF is proportional to the local sensitivity for that interval.
These intervals will naturally be defined by the maximal and minimal values the function can take after changing $\ell$ individual's data.
Unlike smooth sensitivity, there is no smoothing of these scale parameters, allowing for a steeper distribution.

\noindent
\begin{minipage}[t]{0.30\textwidth}
  \vspace{-10pt}%
  \includegraphics[width=\linewidth,trim=0 0 0 0,clip]{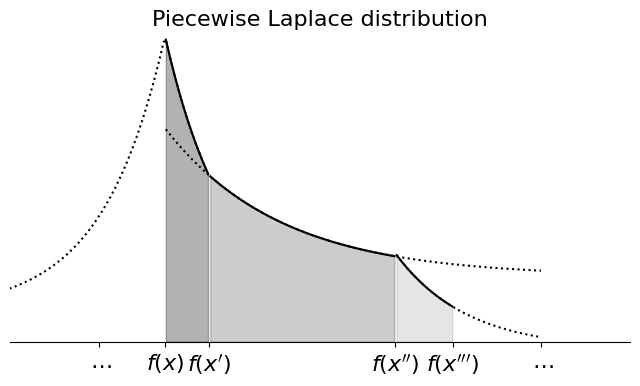}
\end{minipage}
\hspace{0.5em}%
\begin{minipage}[t]{0.68\textwidth}
  \vspace{-12pt}%
Upon initial impression, such a distribution may appear both non-private and difficult to construct.
First, we note that the construction is simplified by the fact that the Laplace distribution is a two-sided exponential distribution and the exponential distribution enjoys a memoryless property.
Next, we note that by setting the scale parameter proportional to the length of the interval (a lower bound on it's local sensitivity),
the total exponential decay over each respective interval will actually be identical.
\end{minipage}

This property is also true of the inverse sensitivity mechanism, where the practical implementation in the continuous setting  uniformly draws from these same intervals at exponentially decaying rates.
Combining these observations, we can construct our desired distribution by sampling an interval through the inverse sensitivity mechanism and then drawing a point in the interval from the truncated exponential distribution with scale proportional to the length of the drawn interval.

\begin{figure}[ht]
  \centering
  \makebox[\textwidth][c]{%
    \includegraphics[width=0.3\textwidth]{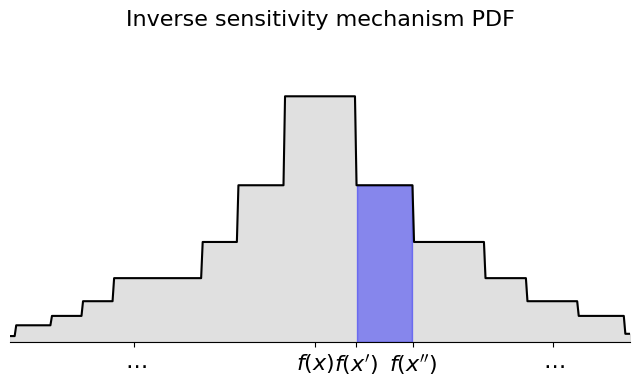}%
    \hspace{30pt}%
    \includegraphics[width=0.3\textwidth]{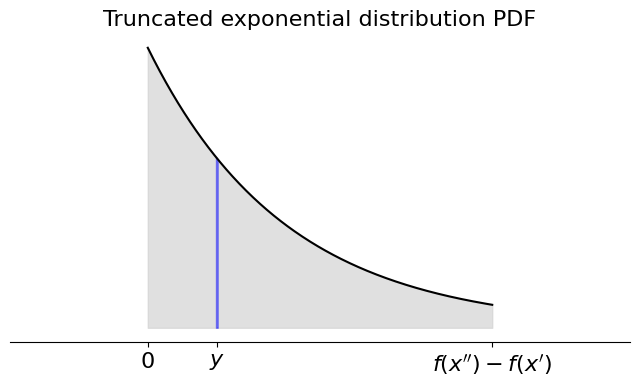}%
  }

  \begin{minipage}{0.8\textwidth}
    \centering
    \caption{\small Two-step sampling procedure for piecewise Laplace distribution}
    \label{fig:sampling_procedure}
  \end{minipage}

\end{figure}

\vspace{-15pt}

Drawing from the truncated exponential distribution instead of the uniform distribution, as is done for inverse sensitivity mechanism, immediately implies improved accuracy using our methodology.
Moreover, this change will come with the exact same privacy guarantees.
It is known that drawing from the exponential distribution can be equivalently achieved by drawing from the exponential mechanism with an appropriate 

\noindent
\begin{minipage}[t]{0.68\textwidth}
  \vspace{-12pt}%
  linear quality score function.
Utilizing this connection, the second step in our sampling procedure  actually creates a
 linear interpolation of the inverse sensitivities for the quality score function 
 because the slope is proportional to the scale parameter.
As a result, our piecewise Laplace mechanism can be identically constructed through a single call to the exponential mechanism.  

\end{minipage}
\hspace{0.5em}%
\begin{minipage}[t]{0.30\textwidth}
  \vspace{-10pt}%
  \includegraphics[width=\linewidth,trim=0 0 0 0,clip]{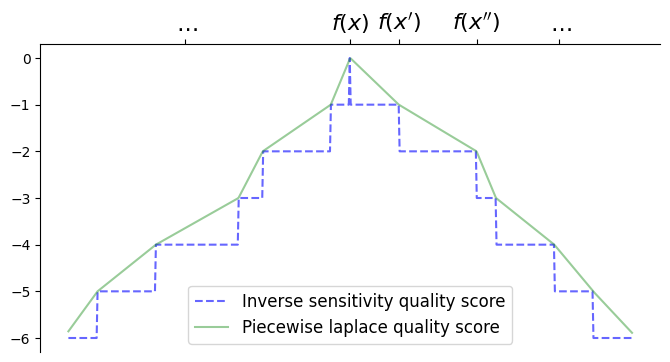}
\end{minipage}

Achieving equivalent privacy guarantees will simply follow from bounding the sensitivity of our piecewise Laplace quality score function which will be relatively straightforward.
Additionally, this interpretation allows for an easy reduction back to Laplace mechanism in the worst-case scenario when all the interval lengths (lower bounds on the local sensitivities) match the global sensitivity.



\subsection{Previous methodology}\label{subsec:previous_method}


The \textit{smooth sensitivity mechanism} in \cite{nissim2007smooth} constructs a smooth sensitivity by sufficiently increasing over all local sensitivities to ensure that any change in sensitivity must be small and adds noise proportionally.
This was later extended to other popular variants of the differential privacy definition \cite{bun2019average}.
The \textit{propose-test-release} framework in \cite{dwork2009differential} instead proposes a sensitivity bound and only releases an estimate if the bound holds for all datasets close to the underlying data.
This was also extended to high dimensional regression problems in \cite{thakurta2013differentially}.
The shortfall of both techniques is that having just one dataset reasonably close to the underlying data with larger local sensitivity causes either much greater noise addition or no estimate release.
Consequently, this has also led to a line of work, including \cite{liu2022differential,asi2023robustness,hopkins2023robustness}, that further explores the connection between privacy and robust estimators.

In contrast, the \textit{inverse sensitivity mechanism}
\footnote{Tracing the exact origin of this framework is challenging as it was considered in the original work of \cite{mcsherry2007mechanism} but not included.
Early instantiations in both the discrete and continuous setting were seen in \cite{smith2011privacy,mir2011pan,johnson2013privacy}.
} 
eschews standard noise addition and employs the exponential mechanism \cite{mcsherry2007mechanism}.
In particular, outcomes are chosen with probability proportional to the number of individuals needing to change their data to match the function output.
This mechanism was fully enumerated in \cite{asi2020near}.
As a result, this methodology overcomes the difficulty of the previous techniques and it's optimality properties were thoroughly investigated in \cite{asi2020near, asi2020instance} showing near instance-optimality with respect to unbiased mechanisms and uniformly outperforming the smooth sensitivity mechanism.

A variety of preprocessing methods, including \cite{chen2013recursive,blocki2013differentially,kasiviswanathan2013analyzing,cummings2020individual}, have also been developed to approximate the function with a restricted sensitivity by which noise can be added proportionally where the approximation is accurate if the underlying data has low sensitivity.
There are also methods that explicitly or inherently incur bias on the estimation for functions that are more difficult to upper bound allowing for lower variance estimates in these settings \cite{fang2022shifted, durfee2024instance}.
The former method, which shifts the inverse sensitivity mechanism for monotonic functions, was then extended to general functions for privacy wrappers that are only given black-box query access to the function of interest \cite{linder2025privately}.

\subsection{Our results}

We unify the Laplace mechanism with near instance optimality by constructing a piecewise Laplace distribution that draws noise proportional to the local sensitivity of each respective interval.
Our \textit{piecewise Laplace mechanism} improves upon the previous widely-used frameworks for exploiting local sensitivity and similarly applies generally.
We actually view our main contribution as the initially counterintuitive but ultimately simple and effective connection between fundamental techniques (detailed in Section~\ref{subsec:our_techniques}).
We will also accompany this with the requisite formalization.

In Section~\ref{sec:cubist_laplace}, we show that the piecewise Laplace mechanism satisfies pure differential privacy and provide a practical sampling procedure.
In Section~\ref{sec:optimality}, we show that the piecewise Laplace mechanism strictly dominates the inverse sensitivity mechanism in the continuous setting and reduces to Laplace mechanism in the worst-case.
In Section~\ref{sec:extensions}, we provide an approximate variant for potentially easing computational costs that can extend to higher dimensions.


\section{Preliminaries}\label{sec:notation}

We will denote the data universe as $\calX$. Our definition of neighboring can be either \textit{swap} or \textit{add-subtract}.

\begin{definition}\label{def:hamming}
    Let $\xx, \xx'$ be datasets of our data universe $\calX$. We define $d(\xx,\xx')$ to be the distance between these two datasets. 
    If $d(\xx, \xx') \leq 1$ then $\xx,\xx'$ are \textit{neighboring} datasets.
    
\end{definition}

\begin{definition}\label{def:dp}\cite{dwork2006calibrating, dwork2006our}
	A mechanism $M:\cX\to \calY$ is $(\diffp, \delta)$-differentially-private (DP) if for any neighboring datasets $\xx,\xx' \in \cX$ and $S \subseteq \calY$:
	\[
		{\Pr[M(\xx) \in S]}\leqslant
		e^{\diffp}{\Pr[M(\xx') \in S]} + \delta.
	\]
 If $\delta = 0$ then $M$ is $\diffp$-DP.
\end{definition}

For a function $f: \cX \to \R$ and dataset $\xx$, the $\diffp$-DP canonical \textit{Laplace mechanism} returns $f(\xx) + \lap(\frac{\Delta}{\diffp})$ where the global sensitivity $\Delta = \sup_{\xx,\xx':d(\xx,\xx')\leq 1} |f(\xx) - f(\xx')|$
and $\lap(b)$ has PDF $p_\lap(z;b) = \frac{1}{2b} \exp\left(-|z|/b \right)$.

\begin{definition}\cite{mcsherry2007mechanism}
The Exponential Mechanism is a randomized mapping $M_q: \cX \to \calY$ with quality score function $q: \calY \times \cX \to \R$ such that 

\[
\prob{}{M_q(\xx) = y} \propto \exp\left(\frac{\diffp \cdot q(y, \xx)}{2}\right) 
\]

\end{definition}

\begin{proposition}\label{prop:exp_mech}\cite{mcsherry2007mechanism}
The exponential mechanism is $\diffp / \Delta_q$-DP where $\Delta_q$ is the sensitivity of $q$
\end{proposition}

\section{Piecewise Laplace Mechanism}\label{sec:cubist_laplace}


In this section, we first construct our piecewise Laplace mechanism through a single call to the exponential mechanism that will provide our privacy guarantees.
Next, we give a practical 2-step sampling procedure and show that it yields an identical distribution.
For more general intuition, we point to Section~\ref{subsec:our_techniques}.

\subsection{Construction through exponential mechanism}\label{subsec:exponential_construction}

In order to construct our piecewise Laplace mechanism through the exponential mechanism, we need to define the appropriate quality score function and then bound the sensitivity of this function.
We first define the upper and lower bounds of the function after changing at most $\ell$ individual's data.

 \[
\widebar{f}(\xx;\ell) \defeq \sup_{\xx'} \{f(\xx')
: d(\xx, \xx') \leq \ell \} 
\text{ \ \ \ \ \ \ and \ \ \ \ \ \ }
\underline{f}(\xx;\ell) \defeq \inf_{\xx'} \{f(\xx')
: d(\xx,\xx') \leq \ell \}
\]

These bounds define the dataset-specific intervals for our piecewise Laplace distribution.
The privacy guarantees will essentially all be a result of the following corollary which is due to the triangle inequality.

\begin{corollary}\label{cor:simple_upper_lower_bounds}
    For any neighboring $\xx,\xx'$ we have $\widebar{f}(\xx;\ell) \leq \widebar{f}(\xx';\ell+1)$ and $\underline{f}(\xx;\ell) \geq \underline{f}(\xx';\ell+1)$
\end{corollary}

For a potential outcome $y \in \R$, we also need to define the minimum number of individuals that need to change their data such that the bounds contain the outcome

    \[
    \elly{y}{\xx} \defeq 
        \inf \{\ell: \underline{f}(\xx;\ell) \leq y \leq \widebar{f}(\xx;\ell) \} 
    \]



\begin{corollary}\label{cor:inv_sens_closeness}
    For any neighboring datasets $\xx,\xx'$ we have $|\elly{y}{\xx} - \elly{y}{\xx'}| \leq 1$ 
    
\end{corollary}

This corollary should technically should be for $y \in [\inf_{\xx}f(\xx), \sup_{\xx}f(\xx)]$ where it follows from the definition of $\elly{y}{\xx}$ and Corollary~\ref{cor:simple_upper_lower_bounds}.

\begin{remark}
    Abusing notation, we could let $\inf \emptyset = \infty$ and $e^{-\infty} = 0$, along with $c + \infty = \infty$ and $c / \infty = 0$ for any $c \in \R$ to handle the possible edge cases of $y \notin [\inf_{\xx}f(\xx), \sup_{\xx}f(\xx)]$ and $\widebar{f}(\xx;\ell) = \infty$ or $\underline{f}(\xx;\ell) = -\infty$.
 But we generally ignore these for simplicity in the exposition.
\end{remark}

Our goal is to construct a quality score function such that outcomes $\widebar{f}(\xx;\ell)$ and $\underline{f}(\xx;\ell)$ have score -$(\ell+1)$ (the first term) and then construct a linear interpolation between these points (the second term).


\begin{definition}\label{def:quality_score}
    Given a function $f:\cX \to \R$ 
    we define our quality score function as 
    \[
    -\qplm{y}{\xx} = \begin{cases}
        \elly{y}{\xx} +
        \frac{y - \widebar{f}(\xx;\elly{y}{\xx}-1)}{\widebar{f}(\xx;\elly{y}{\xx}) - \widebar{f}(\xx;\elly{y}{\xx}-1)} \text{~ ~ ~ ~ ~ ~ ~ if $y > f(\xx)$}
        \\[0.6em]
        \elly{y}{\xx} +
        \frac{\underline{f}(\xx;\elly{y}{\xx}-1) - y}{\underline{f}(\xx;\elly{y}{\xx}-1) - \underline{f}(\xx;\elly{y}{\xx})} \text{~ ~ ~ ~ ~ ~ ~ if $y < f(\xx)$}
        \end{cases}
    \]
    and set $-\qplm{f(\xx)}{\xx} = 1$.

\end{definition}

Note that the quality score function could be discontinuous, for instance if $\widebar{f}(\xx,\ell) = \widebar{f}(\xx,\ell+1)$ for some $\ell$.
Such scenarios commonly occur for median estimation
where the datapoints are equivalent to the $\widebar{f}(\xx,\ell)$ and $\underline{f}(\xx,\ell)$ values.

Our piecewise Laplace mechanism, $\plm$, will  invoke exponential mechanism 
giving the density function

\begin{equation}\label{eq:cubist_laplace}
    \pi_{\plm}(y) = \frac{e^{\qplm{y}{\xx} \diffp / 2}}{\int_{\R} e^{\qplm{y}{\xx}  \diffp / 2} dy} \tag{M.1}
\end{equation}

\begin{theorem}\label{thm:cubist_is_dp_and_br}
    The piecewise Laplace mechanism given in \ref{eq:cubist_laplace} is $\diffp$-DP
\end{theorem}

In Section~\ref{subsec:reduction_laplace} we discuss how we can actually obtain better privacy guarantees that were shown to be inherent to the exponential mechanism in \cite{durfee2019practical}.
Before proving our privacy guarantee, we first show that the linear interpolation in our quality score function must be contained in $[0,1]$ by construction.

\begin{corollary}\label{cor:line_portion_in_zero_one_range}

For $y > f(\xx)$ and $y < f(\xx)$ we respectively have

\[
\frac{y - \widebar{f}(\xx;\elly{y}{\xx}-1)}{\widebar{f}(\xx;\elly{y}{\xx}) - \widebar{f}(\xx;\elly{y}{\xx}-1)} \in [0,1] 
\text{ \ \ \ \ \ \ and \ \ \ \ \ \ }
\frac{\underline{f}(\xx;\elly{y}{\xx}-1) - y}{\underline{f}(\xx;\elly{y}{\xx}-1) - \underline{f}(\xx;\elly{y}{\xx})} \in [0,1] 
\]
    
\end{corollary}

\begin{proof}

By construction, if $y > f(\xx)$ then $\elly{y}{\xx} > 0$ and $y > \widebar{f}(\xx;\elly{y}{\xx}-1)$. Further $y \leq \widebar{f}(\xx;\elly{y}{\xx})$ by definition. The other claim follows symmetrically.
    
\end{proof}

Additionally, we require the following useful fact when comparing line segments.

\begin{fact}\label{fact:helper_for_cubist_dp_thm}
If $u_2 > u_1$ and $v_2 > v_1$, along with $u_1 \leq v_1$ and $u_2 \leq v_2$ then for $z \in [v_1,  v_2]$

 \[
\frac{z - v_1}{v_2 - v_1} \leq \frac{z - u_1}{u_2 - u_1}
\text{ \ \ \ \ \ \ and \ \ \ \ \ \ }
\frac{u_2 - z}{u_2 - u_1} \leq \frac{v_2 - z}{v_2 - v_1}
\]

\end{fact}

\begin{proof}
The first claim can be shown through the inequality chain

\[
\frac{z - u_1}{u_2 - u_1} \geq \frac{z - u_1}{v_2 - u_1} \geq \frac{z - v_1}{v_2 - v_1}
\]

where the first inequality follows from $z \geq u_1$ and $v_2 \geq u_2$. The second inequality reduces to $z(v_1 - u_1) \leq v_2(v_1 - u_1)$ which follows from $v_1 - u_1 \geq 0$ and $z \leq v_2$.
The second claim can be proven equivalently.



\end{proof}

We are now ready to prove that our piecewise Laplace mechanism satisfies differential privacy.
Our construction simply calls the exponential mechanism and thus we only need to appropriately bound the sensitivity of our quality score function.

\begin{proof}[Proof of Theorem~\ref{thm:cubist_is_dp_and_br}]
By Proposition~\ref{prop:exp_mech} it suffices to show $|\qplm{y}{\xx}  - \qplm{y}{\xx'} | \leq 1$ for any neighboring $\xx, \xx'$.
We consider the following sufficient cases.

Case 1: $y \geq f(\xx)$ and $y \leq f(\xx')$. By construction, $\widebar{f}(\xx;1) \geq f(\xx')$ so $\elly{y}{\xx} \leq 1$. Similarly, $\underline{f}(\xx';1) \leq f(\xx)$ so $\elly{y}{\xx'} \leq 1$. If $\elly{y}{\xx} = 0 $ then $y = f(\xx)$ and $-\qplm{y}{\xx} = 1$. Likewise for $\xx'$. Otherwise $\elly{y}{\xx} = 1 $ and/or $\elly{y}{\xx'} = 1 $. Thus our sensitivity bound follows from Corollary~\ref{cor:line_portion_in_zero_one_range}.

Case 2: $y > f(\xx)$ and $y > f(\xx')$. If $\elly{y}{\xx} = \elly{y}{\xx'}$ then our sensitivity bound follows from Corollary~\ref{cor:line_portion_in_zero_one_range}. Otherwise, by Corollary~\ref{cor:inv_sens_closeness} we can assume without loss of generality that $\elly{y}{\xx} + 1 = \elly{y}{\xx'}$. By Corollary~\ref{cor:simple_upper_lower_bounds} we then have $\widebar{f}(\xx;\elly{y}{\xx}) \leq \widebar{f}(\xx';\elly{y}{\xx'})$ and $\widebar{f}(\xx;\elly{y}{\xx}-1) \leq \widebar{f}(\xx';\elly{y}{\xx'}-1)$. It then follows from the first claim of Fact~\ref{fact:helper_for_cubist_dp_thm} that

\[
\frac{y - \widebar{f}(\xx;\elly{y}{\xx}-1)}{\widebar{f}(\xx;\elly{y}{\xx}) - \widebar{f}(\xx;\elly{y}{\xx}-1)} \geq 
\frac{y - \widebar{f}(\xx';\elly{y}{\xx'}-1)}{\widebar{f}(\xx';\elly{y}{\xx'}) - \widebar{f}(\xx';\elly{y}{\xx'}-1)}
\]

Our sensitivity bound then follows from Corollary~\ref{cor:line_portion_in_zero_one_range}.

Case 3: $y < f(\xx)$ and $y < f(\xx')$. Follows symmetrically to case 2.

\end{proof}

\subsection{Sampling procedure}

In this section we give a practical sampling method of our piecewise Laplace mechanism.
Our procedure will first sample an interval proportional to it's length and an exponential decay factor,
and then a point in the interval from the truncated exponential distribution. 

The length of these intervals, referred to as it's local sensitivity in Section~\ref{subsec:our_techniques}, plays a critical role where
 \[
\marginal{\xx}{\ell} \defeq \widebar{f}(\xx;\ell) - \widebar{f}(\xx;\ell-1)
\text{ \ \ \ \ \ \ and \ \ \ \ \ \ }
\marginal{\xx}{-\ell}  \defeq \underline{f}(\xx;\ell-1) - \underline{f}(\xx;\ell)
\]
are the positive and negative marginal sensitivity from changing the $\ell$th individual's data.
In fact, these actually lower bound the local sensitivity at distance $\ell$, i.e. $\max\{\marginal{\xx}{\ell} , \marginal{\xx}{-\ell} \} \leq \sup_{\xx': d(\xx,\xx') = \ell-1} \Delta(\xx')$ where 
$\Delta(\xx') = \max\{\marginal{\xx'}{-1}, \marginal{\xx'}{1}\}$ is the local sensitivity of $\xx'$.

Recall that our goal is to scale the noise of our piecewise Laplace distribution proportional to this lower bound on the respective local sensitivity for each interval.
For this construction we now utilize the fact that the Laplace distribution is simply a two-sided exponential distribution.
The changing of scale parameters in a piecewise manner is further simplified by the fact that the exponential distribution is memoryless.
This then allows us to draw noise from truncated exponential distributions $\expo^{\top}(\Delta,\diffp)$ such that $\diffp, \Delta > 0$ where the PDF is $p_{\expo^{\top}}(z;\Delta,\diffp)$ with

 \[
(1 - e^{-\diffp}) p_{\expo^{\top}}(z;\Delta,\diffp) = \begin{cases}
        \frac{\diffp}{\Delta} \exp\left(-z \cdot \diffp / \Delta \right) \text{~ ~ ~$0 \leq z \leq \Delta$}
        \\
        0 \text{~ ~ ~ ~ ~ ~ ~ ~ ~ ~ ~ ~ ~ ~ ~ otherwise}
        \end{cases}
 \]

In particular, for each interval we will scale the noise drawn from the truncated exponential distribution proportional to it's respective $\marginal{\xx}{\ell}$ or $\marginal{\xx}{-\ell}$.
The fact that this matches the length of the interval implies that the total exponential decay over each interval is actually identical which will allow for our connection to the inverse sensitivity mechanism (formalized in Section~\ref{subsec:dominates_inverse}).


\begin{algorithm}
\caption{Practical sampling procedure for piecewise Laplace mechanism \label{algo:plm_sample}}
\begin{algorithmic}[1]
\REQUIRE Input dataset $\xx$
\STATE Sample interval $\mathbbm{1} \cdot \ell$ where $\mathbbm{1} \in \{1,-1\}$ with probability 
    \[
    \frac{\exp\left(-\ell \cdot \diffp / 2\right)  \marginal{\xx}{\mathbbm{1} \cdot \ell}}
    {\sum_{\ell>0} \exp\left(-\ell \cdot \diffp / 2\right) \left(\marginal{\xx}{\ell} + \marginal{\xx}{-\ell}\right)}
    \] 
\STATE Draw $z \sim \expo^{\top}(\marginal{\xx}{\mathbbm{1} \cdot \ell},\diffp / 2)$.
\RETURN $\widebar{f}(\xx;\ell-1) + z$ if $\mathbbm{1} = 1$ or $\underline{f}(\xx;\ell-1) - z$ if $\mathbbm{1} = -1$
\end{algorithmic}
\end{algorithm}

As is commonly assumed, we will (mostly) need the range of $f$ to be bounded in practice.
Connecting to $q_\texttt{plm}$ in Section~\ref{subsec:exponential_construction}, step 1 of this sampling procedure, which mimics the practical construction of the inverse sensitivity mechanism, provides the first term of the quality score function.
Step 2 of the procedure then provides the second term, connecting the truncated exponential distribution with the linear interpolation, in which the slope is proportional to the respective interval length.

\begin{lemma}
Let $M$ denote the mechanism in Algorithm~\ref{algo:plm_sample}, then $M \equiv \plm$ where $\plm$ is our piecewise Laplace mechanism in \ref{eq:cubist_laplace}.

\end{lemma}

\begin{proof}
By construction of Algorithm~\ref{algo:plm_sample} and applying our definitions in Section~\ref{subsec:exponential_construction}, the density function of our sampling procedure is

    \[
    \pi_{M}(y) = \begin{cases}
    \frac{\exp\left(-\elly{y}{\xx} \cdot \diffp / 2\right) \marginal{\xx}{\elly{y}{\xx}}}
    {\sum_{\ell>0} \exp\left(-\ell \cdot \diffp / 2\right) \left(\marginal{\xx}{\ell} + \marginal{\xx}{-\ell}\right)}
    \left( \frac{\diffp\exp\left(-(y -  \widebar{f}(\xx;\elly{y}{\xx}-1)) \cdot \diffp / 2\marginal{\xx}{\elly{y}{\xx}} \right)}
    {2\marginal{\xx}{\elly{y}{\xx}} (1 - e^{-\diffp/2})}   \right)
        \text{~ ~ ~ ~ ~ ~ if $y \geq f(\xx)$}
        \\[0.6em]
    \frac{\exp\left(-\elly{y}{\xx} \cdot \diffp / 2\right) \marginal{\xx}{-\elly{y}{\xx}}}
    {\sum_{\ell>0} \exp\left(-\ell \cdot \diffp / 2\right) \left(\marginal{\xx}{\ell} + \marginal{\xx}{-\ell}\right)}
    \left( \frac{\diffp\exp\left(-(\underline{f}(\xx;\elly{y}{\xx}-1) - y) \cdot \diffp / 2\marginal{\xx}{-\elly{y}{\xx}} \right)}
    {2\marginal{\xx}{-\elly{y}{\xx}} \cdot (1 - e^{-\diffp/2})}   \right)
        \text{ ~ ~ ~ ~ if $y < f(\xx)$}
        \end{cases}
    \]

In order to show equivalence, it suffices to show that for any $y, y'$ we have 

\[
\frac{\pi_{M}(y)}{\pi_{M}(y')} = \frac{\pi_{\plm}(y)}{\pi_{\plm}(y')}
\]

We'll assume $y \geq f(\xx)$ and $y' < f(\xx)$ and the other cases will follow equivalently. This choice will also make it easier to visualize the following cancellation of like terms from our provided density function

\[
\frac{\pi_{M}(y)}{\pi_{M}(y')} =
\frac{\exp\left(-\elly{y}{\xx} \cdot \diffp / 2\right) \cdot \exp\left(-(y -  \widebar{f}(\xx;\elly{y}{\xx}-1)) \cdot \diffp / 2\marginal{\xx}{\elly{y}{\xx}} \right)}
{\exp\left(-\elly{y'}{\xx} \cdot \diffp / 2\right) \cdot \exp\left(-(\underline{f}(\xx;\elly{y'}{\xx}-1) - y') \cdot \diffp / 2\marginal{\xx}{-\elly{y'}{\xx}} \right)}
\]

Applying our quality score function in Definition~\ref{def:quality_score} then gives our desired ratio

\[
\frac{\pi_{M}(y)}{\pi_{M}(y')} =
\frac{e^{\qplm{y}{\xx} \diffp / 2}}{e^{\qplm{y'}{\xx} \diffp / 2}}
\]

\end{proof}

\section{Optimality}\label{sec:optimality}

In this section, we first show that in the worst-case where all local sensitivities equal the global sensitivity, our piecewise Laplace mechanism simply reduces to a Laplace mechanism.
This completes the equivalency between our exponential mechanism construction, our sampling procedure, and our piecewise Laplace distribution with scale parameters proportional to the local sensitivity for the interval.
Additionally, we prove that in the continuous setting our method strictly dominates the inverse sensitivity mechanism, which was previously shown to be nearly instance optimal with respect to unbiased mechanisms and uniformly outperform smooth sensitivity mechanism.
A smooth variant of the inverse sensitivity was also previously introduced and we show that our method can extend to this variant as well, although smoothing isn't generally necessary for our mechanism.

\subsection{Reduction to Laplace mechanism}\label{subsec:reduction_laplace}

We consider the worst-case scenario in which all local sensitivities equal the global sensitivity. 
More formally, suppose for some dataset $\xx$ we have $\widebar{f}(\xx;\ell+1) - \widebar{f}(\xx;\ell) = \Delta$ and $\underline{f}(\xx;\ell) - \underline{f}(\xx;\ell+1) = \Delta$ for all $\ell \geq 0$. This implies $\widebar{f}(\xx;\ell) - f(\xx) = \ell \cdot \Delta$ and $f(\xx) - \underline{f}(\xx;\ell) = \ell \cdot \Delta$.
It then follows that for dataset $\xx$ we have $\qplm{y}{\xx} = -\frac{|f(\xx) - y|}{\Delta} - 1$.
The exponential mechanism is invariant under additive shifts and exponential mechanism with $q_{\texttt{lap}}(y;\xx) = -\frac{|f(\xx) - y|}{\Delta}$ is identical to Laplace mechanism  $f(\xx) + \lap(\frac{2\Delta}{\diffp})$ \cite{mcsherry2007mechanism}.
Therefore, our piecewise Laplace mechanism reduces to a Laplace mechanism for this worst-case dataset $\xx$ as desired.
Additionally, if the function range is bounded (strictly necessary for \textit{swap} neighboring in order to have bounded global sensitivity) then it becomes a truncated Laplace distribution in that range.


It is important to note that even though $q_{\texttt{lap}}(y;\xx) = -\frac{|f(\xx) - y|}{\Delta}$ has sensitivity 1, this alternate instantiation of Laplace mechanism through exponential mechanism instead achieves $\diffp/2$-DP because it's a special symmetric variant.
In particular, \cite{durfee2019practical,dong2020optimal}
showed that the privacy loss of the exponential mechanism can be tightly characterized instead by a \textit{bounded range} (BR) property.

\begin{proposition}\cite{durfee2019practical}\label{prop:exp_range_bounded}
    Given exponential mechanism $M_q$ and neighboring datasets $\xx, \xx'$ we have 

    \[
    D_{\infty}(M_q(\xx)|| M_q(\xx')) + D_{\infty}(M_q(\xx')|| M_q(\xx)) \leq \diffp / \Delta_q 
    \]

    where the Renyi divergence $D_\infty(M(\xx)||M(\xx')) \defeq \sup_y \log(\Pr[M(\xx) = y]/\Pr[M(\xx')=y])$
\end{proposition}

Pure DP can equivalently be defined as $D_\infty(M(\xx)||M(\xx')) \leq \diffp$ for all neighboring $\xx,\xx'$, but Renyi divergence doesn't necessarily satisfy symmetry.
\footnote{The bounded-range definition satisfies symmetry by construction. Generally for Renyi divergence we can induce a distance metric with a slight relaxation $(D_\alpha(M(\xx)|| M(\xx') + D_\alpha(M(\xx')|| M(\xx))/2$ that still preserves privacy properties, such as composition.}
This means that in general Proposition~\ref{prop:exp_range_bounded} only implies $\diffp / \Delta_q$-DP, matching Proposition~\ref{prop:exp_mech}.
However, the neighboring quality score symmetry in the instantiation of Laplace mechanism does provide symmetry which then implies $\diffp/2$-DP, matching it's classical formulation.

While the flexibility of our piecewise Laplace mechanism will not always allow for this same symmetry, the concentration properties of bounded range were further studied and it was shown that composing $\diffp$-BR mechanisms gave very nearly the same privacy guarantees as optimally composing $\diffp/2$-DP mechanisms (see Figure 1 in \cite{dong2020optimal}).
Additionally, for another commonly-used privacy definition, zero-Concentrated Differential Privacy (zCDP), $\diffp$-BR implies $\frac{1}{8}\diffp^2$-zCDP \cite{cesar2021bounding}, while $\diffp$-DP implies $\frac{1}{2}\diffp^2$-zCDP \cite{bun2016concentrated}, which is very nearly tight for $\diffp \leq 1$.
Consequently, even in the worst-case sensitivity scenario where the Laplace mechanism can be considered tight for DP, our piecewise Laplace mechanism adds the same noise and very nearly matches it's privacy guarantees under DP composition and zCDP measurement.

\subsection{Comparison to inverse sensitivity mechanism}\label{subsec:dominates_inverse}

In this section we show that our piecewise Laplace mechanism strictly dominates the inverse sensitivity mechanism in the continuous setting.
Recall that both mechanisms achieve the same privacy guarantees.
Here we will formally connect the first step in our sampling procedure to the sampling from inverse sensitivity mechanism in the continuous setting.
Thus, both mechanisms require the same computational complexity and we will show that the accuracy of our mechanism is strictly better.
We first present the inverse sensitivity mechanism with it's requisite definitions.

\begin{definition}\label{def:exact_inverse_sensitivity}
    For a function $f: \calX \to \calY$ and $\xx \in \cX$, let the $\mathrm{inverse \ sensitivity}$ of $y \in \calY$ be

    \[
    \len{f}(y;\xx) \defeq \inf_{\xx'} \{ d(\xx, \xx') | f(\xx') = y \}
    \]
\end{definition}

The \textit{inverse sensitivity mechanism} then draws from the exponential mechanism instantiated upon the distance metric giving the density function

\begin{equation}\label{eq:inverse_sensitivity_mech}
    \pi_{M_{\texttt{inv}}} (y) = \frac{e^{-\len{f}(y;\xx) \diffp / 2}}{\int_{\calY} e^{-\len{f}(y;\xx) \diffp / 2} dy} \tag{M.2}
\end{equation}

While the general intuition is that smaller $\len{f}(y;\xx)$ implies  $y$ is closer to $f(\xx)$, it's certainly possible to construct functions by which this does not hold.
However, almost all functions of interest will follow this intuition.
As a result, \cite{asi2020near} defined a general class of functions that doesn't allow for this edge case and was shown to include all continuous functions from a convex domain.

\begin{definition}\label{def:sample_monotone}
Let $f: \calX \to \R$. Then $f$ is \textrm{sample-monotone} if for every $\xx \in \calX$ and $s,t \in \R$ satisfying $f(\xx) \leq s \leq t$ or $t \leq s \leq f(\xx)$, we have $\len{f}(s;\xx) \leq \len{f}(t;\xx)$
    
\end{definition}

Applying this general condition on the functions of interest matches our definitions to the inverse sensitivity.

\begin{corollary}\label{cor:monotone_implies_same_len}
Given a sample-monotone function $f: \cX \to \R$,  we have
    $\len{f}(y;\xx) = \elly{y}{\xx}$ for all $y$
\end{corollary}

\begin{proof}
We consider the case $y \geq f(\xx)$, and $y \leq f(\xx)$ will follow symmetrically.

If we suppose $\len{f}(y;\xx) < \elly{y}{\xx}$ then there exists $\xx'$ such that $f(\xx') = y$ and $d(\xx,\xx') = \len{f}(y;\xx)$. Therefore $y \leq \widebar{f}(\xx, \len{f}(\xx;y))$ which implies $\elly{y}{\xx} \leq \len{f}(y;\xx)$ giving a contradiction.

Next suppose $\len{f}(y;\xx) > \elly{y}{\xx}$. By construction of $\widebar{f}(\xx,\ell)$ there must exist $y'$ such that $f(\xx') = y' \geq y$ and $d(\xx,\xx') \leq \elly{y}{\xx}$. Therefore, $\len{f}(y';\xx) \leq \elly{y}{\xx}$ which contradicts the sample-monotone definition.

\end{proof}

With this equivalence between our definition and the inverse sensitivity, we can connect our sampling procedure to the sampling procedure for the inverse sensitivity mechanism.
As a result we can show that for any distance $\alpha$ from $f(\xx)$, our mechanism is at least as likely to draw a sample within this distance.
Within the proof, we'll see that this inequality is actually strict for almost all $\alpha$.

\begin{theorem}
    Given a sample-monotone function $f: \cX \to \R$, then for all $\alpha$

\[
\Pr[|\plm(\xx) - f(\xx)| \leq \alpha] \geq \Pr[|M_{\texttt{inv}}(\xx) -f(\xx)| \leq \alpha]
\]

\end{theorem}

\begin{proof}
By Corollary~\ref{cor:monotone_implies_same_len} and the definition of mechanism \ref{eq:inverse_sensitivity_mech}, we have that $M_{\texttt{inv}}$ can be equivalently defined by applying our Algorithm~\ref{algo:plm_sample} but instead of drawing from the truncated exponential distribution, simply uniformly sampling over the drawn interval.
In fact, this is how the inverse sensitivity mechanism is constructed in practice for the continuous setting.

It then suffices to show that drawing from the truncated exponential distribution will be more likely to be within a certain distance from $f(\xx)$ for any distance.
Given that the PDF of the uniform distribution over an interval of length $\Delta$ is simply $1/\Delta$, it suffices to show that
$\int_0^y p_{\expo^{\top}}(z;\Delta,\diffp)dz \geq y / \Delta$ for any $0 \leq y \leq \Delta$ and $\diffp > 0$.

Applying the CDF of the exponential distribution this reduces to showing $\frac{1 - e^{-y \cdot \diffp / \Delta}}{1 - e^{-\diffp}} \geq y / \Delta$.
First we see that $\frac{1 - e^{-y \cdot \diffp / \Delta}}{1 - e^{-\diffp}} = y / \Delta$ for $y = 0, \Delta$.
Next we have 

\[
\frac{d^2}{dy^2} \frac{1 - e^{-y \cdot \diffp / \Delta}}{1 - e^{-\diffp}}
=
\frac{-(\diffp / \Delta)^2e^{-y \cdot \diffp / \Delta}}{1 - e^{-\diffp}} < 0
\]

implying strict concavity which gives $\frac{1 - e^{-y \cdot \diffp / \Delta}}{1 - e^{-\diffp}} > y / \Delta$ for $y \in (0, \Delta)$ as desired.

\end{proof}

In addition, it is important to discuss another variant of the inverse sensitivity that was introduced in \cite{asi2020near} which they termed the $\rho$-smooth version and is defined as

\[
    \len{f}^{\rho}(y;\xx) = \inf_{y' \in \calY: \norm{y'-y} \leq \rho} \len{f}(y';\xx)
\]

This variant adds more probability mass within $\rho$ distance of $f(\xx)$ at the cost of spreading the remaining

\noindent
\begin{minipage}[t]{0.68\textwidth}
  \vspace{-12pt}%
probability mass.
The general rule of thumb given in \cite{asi2020near} is to set $\rho = 1 / \texttt{poly}(n) <\!\!< 1 / \sqrt{n}$ where $n$ is the number of individuals in the dataset (assuming \textit{swap} neighboring).
We can adjust our method to this variant by setting $\widebar{f}^{\rho}(\xx;\ell) = \widebar{f}(\xx;\ell) + \rho$ and $\underline{f}^{\rho}(\xx;\ell) = \underline{f}(\xx;\ell) - \rho$ for all $\ell$.
Essentially it just shifts both sides of $f(\xx)$ outwards by $\rho$ and keeps the quality scores constant in the neighborhood of distance $\rho$ around $f(\xx)$.

\end{minipage}
\hspace{0.5em}%
\begin{minipage}[t]{0.30\textwidth}
  \vspace{-12pt}%
  \includegraphics[width=\linewidth,trim=0 0 0 0,clip]{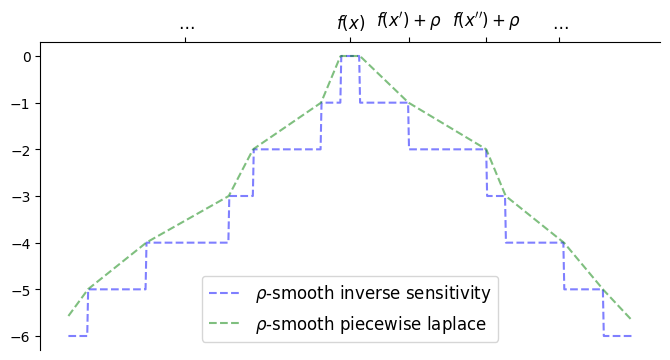}
\end{minipage}

Proving this variant for our method will follow from the fact that now $\qplm{y}{\xx} = q_{\texttt{plm}}^{\rho}(y+\rho;\xx)$ for $y > f(\xx)$ and $\qplm{y}{\xx} = q_{\texttt{plm}}^{\rho}(y-\rho;\xx)$ for $y < f(\xx)$ along with slightly more case analysis when $y \in [f(\xx) - \rho, f(\xx) + \rho]$.
However, for our method this smoothing is generally unnecessary as our quality score function naturally smoothly decays as it moves away from $f(\xx)$.


    

\section{Approximate Extensions}\label{sec:extensions}

In this section we provide an approximate variant of our mechanism that can ease the cost of computing inverse sensitivity while still achieving utility improvements over global sensitivity.
Observe that the primary computational cost of applying our method and inverse sensitivity mechanism is computing the $\widebar{f}$ and $\underline{f}$ from Section~\ref{subsec:exponential_construction}.
As was originally shown in \cite{nissim2007smooth}, these values can be intractable for some functions, however we can still most often apply approximations efficiently.


The general idea to these approximations is straightforward. 
As long as the intervals are at worst further spread out in a consistent manner, so the distribution is consistently flatter, then privacy is maintained.
Essentially, we just need intervals that will still maintain an analogous Corollary~\ref{cor:simple_upper_lower_bounds}.
In this way we can also consider the Laplace mechanism to simply be an approximate variant of our approach as the intervals will all have the length of the global sensitivity.
As such, giving approximate bounds below the global sensitivity at lower computational costs can still effectively improve utility with these techniques.

\subsection{Approximate variant}

While there are a variety of ways we could define the approximation, we will follow equivalently to an approximate variant of inverse sensitivity mechanism introduced in \cite{asi2020instance}.
This variant allows for simple connections to local sensitivity and can also extend to higher dimensions.
We first define bounds on how much the radius of an $L_p$ ball around $f(\xx)$ can increase as we change additional individual's data.

\begin{definition}\label{def:radius_bounding}
    Given a function $f: \cX \to \R^d$ we have radius bounding functions $R_\ell:\cX \to \R$ for all $\ell > 0$, if for any neighboring $\xx,\xx'$ we have $R_{\ell}(\xx) \leq R_{\ell+1}(\xx')$ and $R_1(\xx) \geq \norm{f(\xx) - f(\xx')}_p$

\end{definition}

We also define $\cR_{\ell}(\xx) \defeq \sum_{k \leq \ell}R_k(\xx)$.
These properties ensure that an $L_p$ ball of radius $\cR_{\ell}(\xx)$ around $f(\xx)$ is fully contained in an $L_p$ ball of radius $\cR_{\ell+1}(\xx')$ around $f(\xx')$, which is analogous to our essential Corollary \ref{cor:simple_upper_lower_bounds} but extends to higher dimensions.
We can now define the number of individuals that need to change their data such that the $L_p$ ball with radius $\cR_{\ell}(\xx)$ contains a given outcome $y$.

\[
\widetilde{\ell}(y;\xx) \defeq \inf\{\ell: \cR_{\ell}(\xx) \geq \norm{y - f(\xx)}_p \}
\]

\begin{lemma}\label{lem:approximate_close}
    Given a function $f: \cX \to \R^d$ and radius bounding functions $R_\ell:\cX \to \R$ for all $\ell > 0$, for any neighboring $\xx,\xx'$ we have $|\widetilde{\ell}(y;\xx) - \widetilde{\ell}(y;\xx')| \leq 1$
\end{lemma}

\begin{proof}

Without loss of generality, suppose $\widetilde{\ell}(y;\xx') > \widetilde{\ell}(y;\xx) + 1$. 
By construction $\cR_{\widetilde{\ell}(y;\xx)+1}(\xx') < \norm{y - f(\xx')}_p$. 
Further, by Definition~\ref{def:radius_bounding} we have $\cR_{\ell}(\xx) \leq \sum_{1 < k \leq \ell+1} R_{k}(\xx')$ implying $R_1(\xx') + \cR_{\widetilde{\ell}(y;\xx)}(\xx) \leq \cR_{\widetilde{\ell}(y;\xx)+1}(\xx') $.
Thus $\cR_{\widetilde{\ell}(y;\xx)}(\xx) < \norm{y - f(\xx')}_p - \norm{f(\xx) - f(\xx')}_p$ which by the triangle inequality implies $\cR_{\widetilde{\ell}(y;\xx)}(\xx) < \norm{y - f(\xx)}_p$ and this gives our contradiction.
    
\end{proof}

With these definitions in hand along with facts about their closeness on neighboring datasets, we can now introduce the quality score function of this approximate variant

\[
-\widetilde{q}_{\texttt{plm}}(y;\xx) \defeq \widetilde{\ell}(y;\xx) + \frac{\norm{y-f(\xx)}_p - \cR_{\widetilde{\ell}(y;\xx)-1}(\xx)}{\cR_{\widetilde{\ell}(y;\xx)}(\xx) - \cR_{\widetilde{\ell}(y;\xx)-1}(\xx)}
\]

As in Section~\ref{subsec:exponential_construction}, we will simply invoke the exponential mechanism on this quality score.
If we can bound the sensitivity of the quality score function by 1, then this will ensure $\diffp$-DP by Proposition~\ref{prop:exp_mech}.

\begin{lemma}
Given a function $f: \cX \to \R^d$ and radius bounding functions $R_\ell:\cX \to \R$ for all $\ell > 0$, for any neighboring $\xx,\xx'$ we have
 $|\widetilde{q}_{\texttt{plm}}(y;\xx) - \widetilde{q}_{\texttt{plm}}(y;\xx)| \leq 1$
    
\end{lemma}

\begin{proof}

By construction we have $\cR_{\widetilde{\ell}(y;\xx)} \geq \norm{y - f(\xx)}_p > \cR_{\widetilde{\ell}(y;\xx)-1}$ and thus

\[
\frac{\norm{y-f(\xx)}_p - \cR_{\widetilde{\ell}(y;\xx)-1}(\xx)}{\cR_{\widetilde{\ell}(y;\xx)}(\xx) - \cR_{\widetilde{\ell}(y;\xx)-1}(\xx)} \in [0,1]
\]

This immediately implies $|\widetilde{q}_{\texttt{plm}}(y;\xx) - \widetilde{q}_{\texttt{plm}}(y;\xx)| \leq 1$ if $\widetilde{\ell}(y;\xx) = \widetilde{\ell}(y;\xx')$.
By Lemma~\ref{lem:approximate_close} we now only need to show this when $\widetilde{\ell}(y;\xx) + 1 = \widetilde{\ell}(y;\xx')$.
Definition~\ref{def:radius_bounding}  implies $R_1(\xx') + \cR_{\ell}(\xx) \leq \cR_{\ell+1}(\xx')$ and thus

\[
\frac{\norm{y-f(\xx')}_p - \cR_{\widetilde{\ell}(y;\xx)}(\xx')}{\cR_{\widetilde{\ell}(y;\xx)+1}(\xx') - \cR_{\widetilde{\ell}(y;\xx)}(\xx')} 
\leq
\frac{\norm{y-f(\xx')}_p - (R_1(\xx') + \cR_{\widetilde{\ell}(y;\xx)-1}(\xx))}{\cR_{\widetilde{\ell}(y;\xx)}(\xx) - \cR_{\widetilde{\ell}(y;\xx)-1}(\xx)} 
\leq 
\frac{\norm{y-f(\xx)}_p -  \cR_{\widetilde{\ell}(y;\xx)-1}(\xx)}{\cR_{\widetilde{\ell}(y;\xx)}(\xx) - \cR_{\widetilde{\ell}(y;\xx)-1}(\xx)} 
\]

where the first inequality follows by Fact~\ref{fact:helper_for_cubist_dp_thm} and the second follows from the triangle inequality as $R_1(\xx') \geq \norm{f(\xx) - f(\xx')}_p$.
We previously showed that these fractions are in the range $[0,1]$ which then implies $|\widetilde{q}_{\texttt{plm}}(y;\xx) - \widetilde{q}_{\texttt{plm}}(y;\xx)| \leq 1$ as desired.
    
\end{proof}

In order to sample from this approximate variant for the 1-dimensional setting, this will just translate to using $\widebar{f}(\xx;\ell) = f(\xx) + \cR_{\ell}(\xx)$ and $\underline{f}(\xx;\ell) = f(\xx) - \cR_{\ell}(\xx)$ in our Algorithm~\ref{algo:plm_sample}.
For higher dimensions this will become much more computationally challenging but still possible with closed form solutions for $L_1$ balls through the incomplete gamma function.
However, this will still require a rejection sampling procedure in the second step when we sample from the high-dimensional Laplacian until it lands in the correct $L_1$ ball (and not in others) that could become expensive in high dimensions.
Although the full construction would likely still be bottle-necked by computing decent approximations of inverse sensitivities in higher dimensions, improving upon this sampling would be an interesting future direction.

\subsection{Lowering computational costs}

The approximate variant can be particularly useful if computing all of the different sensitivity bounds is challenging, ie the $\widebar{f}$ and $\underline{f}$ from Section~\ref{subsec:exponential_construction}.
In addition, we could still maintain the asymmetry and relax this approximate technique in the 1-dimensional setting by using techniques from \cite{durfee2024instance}.

We first suppose that it's much simpler computationally to compute local sensitivities and show that any upper bound on these sensitivities applies as the radius bounding functions (equivalent to \cite{asi2020instance}).

\begin{corollary}\label{cor:local_for_approx}
    Given local sensitivity upper bounds $\widetilde{\Delta}(\xx) \geq \sup_{\xx':d(\xx,\xx')\leq 1} \norm{f(\xx) - f(\xx')}_p$, setting $R_{\ell}(\xx) = \sup_{\xx':d(\xx,\xx')\leq \ell} \widetilde{\Delta}(\xx')$ satisfies Definition~\ref{def:radius_bounding}.
\end{corollary}

\begin{proof}
    The condition that $R_1(\xx) \geq \norm{f(\xx) - f(\xx')}_p$ follows by construction. The condition that $R_{\ell}(x) \leq R_{\ell+1}(\xx')$ follows from the fact that for any dataset $\xx''$ if $d(\xx, \xx'') \leq \ell$ then $d(\xx', \xx'') \leq \ell+1$ because $\xx$ and $\xx'$ are neighbors.
\end{proof}

A more important variant that can ease the computation even more significantly is by setting the local sensitivities to the global sensitivity once a sufficient distance from $\xx$ has been reached.
In particular, we can consider some fixed distance $k$, and instead set $R_{\ell}(\xx) = \Delta$ for $\ell \geq k$ in Corollary~\ref{cor:local_for_approx}.
For sufficient value of $k$ it becomes incredibly unlikely that these intervals will be selected from and setting them to the global sensitivity will have a negligible effect on accuracy.
Setting $k = O(1/\diffp)$ should be sufficient, which significantly improves runtime as we only need compute local sensitivity bounds within $k$ distance of $\xx$. 
Other approximate variants could simply have the interval after $k$ go out to the upper or lower bound of the function, thus we'd only sample from $O(1/\diffp)$ intervals even if the dataset is large.
There are many options and essentially all these approximate variants just need to maintain a corresponding Corollary~\ref{cor:simple_upper_lower_bounds}.

We also consider the setting where the global sensitivity of the local sensitivity is nicely bounded, which is true of triangle counting in graphs, and allows for fast implementations of smooth sensitivity and propose-test-release applications seen in \cite{nissim2007smooth, vadhan2017complexity}.
We show that we can use the exact same tricks to get efficient implementations.

\begin{corollary}
Let $\Delta(\xx) = \sup_{\xx':d(\xx,\xx')\leq 1}\norm{f(\xx) - f(\xx')}_p$ and $\Delta' = \sup_{\xx,\xx':d(\xx,\xx')\leq 1}|\Delta(\xx) - \Delta(\xx')|$. Setting $R_{\ell}(\xx) =  \Delta(\xx) + (\ell - 1)\Delta'$ satisfies Definition~\ref{def:radius_bounding}.
    
\end{corollary}

\begin{proof}
    The condition that $R_1(\xx) \geq \norm{f(\xx) - f(\xx')}_p$ follows by construction. Now it suffices to show $\Delta(\xx) + (\ell - 1)\Delta' \leq \Delta(\xx') + \ell\cdot\Delta'$ for neighboring $\xx,\xx'$ which follows from our definition that $\Delta' \geq |\Delta(\xx) - \Delta(\xx')|$.
    
\end{proof}

\section*{Future Directions}\label{sec:future}

Our methods could similarly be extended to other procedures that call exponential mechanism with integral quality score functions by creating advantageous linear interpolations.
For example, it should be straightforward to extend these techniques to the shifted inverse sensitivity mechanism in \cite{fang2022shifted}.
This mechanism was also utilized in \cite{linder2025privately} and it's likely our methods could apply there as well, although it could only improve practical implementations and not their tight asymptotic analysis.

This piecewise methodology could also potentially be extended to the Gaussian mechanism under zCDP.
While the general intuition would still be similar and we can construct the Gaussian mechanism through exponential mechanism with quality score function $-(y-f(\xx))^2 / 2\Delta^2$, there would be significant challenges in the analysis.
Among other difficulties, the Gaussian distribution does not enjoy the memoryless property, i.e. quadratic interpolations of quality score functions is much more messy, and the CDF does not have an analytical closed-form making piecewise integration much more difficult.


\section*{Acknowledgments}

We thank Richard Peng for helpful comments on the exposition.

\bibliographystyle{alpha}
\bibliography{references}



\end{document}